\documentclass[a4paper,10pt,reqno]{amsart}
\usepackage{amsopn}
\usepackage{color}
\usepackage[square,numbers,comma,sort&compress]{natbib}
\usepackage{enumerate}
\usepackage{colortbl}
\usepackage{booktabs}
\usepackage{algorithmic}
\usepackage{algorithm}
\usepackage{ifpdf}
\usepackage{graphics}
\usepackage{epsfig}
\usepackage{subfigure}
\usepackage{booktabs}
\usepackage{mathrsfs}  
\usepackage{setspace}
\usepackage{amsmath}
\usepackage{amssymb}
\usepackage[breaklinks=true,colorlinks,citecolor=blue]{hyperref}
\usepackage{threeparttable}
\usepackage{morefloats}

\usepackage{natbib}

\def\ve#1{\mathchoice{\mbox{\boldmath$\displaystyle\bf#1$}}
{\mbox{\boldmath$\textstyle\bf#1$}}
{\mbox{\boldmath$\scriptstyle\bf#1$}}
{\mbox{\boldmath$\scriptscriptstyle\bf#1$}}}

\newcommand\Z{\mathbb Z}  
  
\newcommand\R{\mathbb R}

\DeclareMathOperator{\IDval}{IDval}    
\DeclareMathOperator{\IDvalInit}{IDvalInit}

\let\epsilon=\varepsilon
\usepackage{ifthen}
\makeatletter
\newcommand{\DeclareBracket}[3]{
  \newcommand{#1}[2][]{%
  \ifthenelse%
  {\equal{##1}{}}%
  {\left#2##2\right#3}%
  {\csname ##1l\endcsname#2##2\csname ##1r\endcsname#3}}}    
\makeatother
\DeclareBracket\abs||           
\DeclareBracket\norm\|\|        
\DeclareBracket\floor\lfloor\rfloor
\DeclareBracket\ceil\lceil\rceil
\DeclareBracket\set\{\}
\DeclareBracket\paren()
\DeclareBracket\inner\langle\rangle
\DeclareBracket\fractional\{\}
\definecolor{purple}{RGB}{160,32,240}       
\definecolor{darkyellow}{RGB}{190,190,0}       
\definecolor{darkgreen}{RGB}{0,120,0}       
\definecolor{darkblue}{RGB}{0,0,180}       

 \makeatletter
 \newtheorem{lemma}{Lemma}
 \renewcommand*{\c@lemma}{\c@theorem}
 \renewcommand*{\p@lemma}{\p@theorem}

 \renewcommand*{\c@conjecture}{\c@theorem}
 \renewcommand*{\p@conjecture}{\p@theorem}

 \newtheorem{proposition}{Proposition}
 \renewcommand*{\c@proposition}{\c@theorem}
 \renewcommand*{\p@proposition}{\p@theorem}

 \renewcommand*{\c@corollary}{\c@theorem}
 \renewcommand*{\p@corollary}{\p@theorem}

 \renewcommand*{\c@observation}{\c@theorem}
 \renewcommand*{\p@observation}{\p@theorem}

 \theoremstyle{definition}
 \newtheorem{problem}{Problem}
 \renewcommand*{\c@problem}{\c@theorem}
 \renewcommand*{\p@problem}{\p@theorem}

 \newtheorem{definition}{Definition}
 \renewcommand*{\c@definition}{\c@theorem}
 \renewcommand*{\p@definition}{\p@theorem}

 \renewcommand*{\c@remark}{\c@theorem}
 \renewcommand*{\p@remark}{\p@theorem}

 \newtheorem{example}{Example}
 \renewcommand*{\c@example}{\c@theorem}
 \renewcommand*{\p@example}{\p@theorem}

\renewcommand*{\c@algorithm}{\c@theorem}
\renewcommand*{\p@algorithm}{\p@theorem}

\title[QuickLexSort]{QuickLexSort: An efficient algorithm for lexicographically sorting nested restrictions of a database}
\author[]{David Haws\\IBM T.J. Watson Research Center, 1101 Kitchawan Road, Yorktown Heights, NY 10598, USA}
\email{dhaws@us.ibm.com, dchaws@gmail.com}
\date{}
\begin{document}


\begin{abstract}
Lexicographical sorting is a fundamental problem with applications to
contingency tables, databases, Bayesian networks, and more. A standard method
to lexicographically sort general data is to iteratively use a stable sort -- a
sort which preserves existing orders. Here we present a new method of
lexicographical sorting called QuickLexSort.  Whereas a stable sort based
lexicographical sorting algorithm operates from the least important to most
important features, in contrast, QuickLexSort sorts from the most important to
least important features, refining the sort as it goes.  QuickLexSort first
requires a one-time modest pre-processing step where each feature of the data set is
sorted independently.  When lexicographically sorting a database, QuickLexSort
(including pre-processing) has comparable running time to using a stable sort
based approach.  For a data base with $m$ rows and $n$ columns, and a sorting
algorithm running in time $O(mlog(m))$, a stable sort based lexicographical
sort and QuickLexSort will both take time $O(nmlog(m))$.  However in many
applications one has the need to lexicographically sort nested data, e.g.\ all
possible sub-matrices up to a certain cardinality of columns. In such cases we
show QuickLexSort gives a performance improvement of a log factor
of the database length (rows in matrix) over using a standard stable sort based
approach. E.g.\ to sort all sub-matrices up to cardinality $k$, QuickLexSort
has running time $O(mn^k)$ whereas a stable sort based lexicographical sort
will take time $O(mlog(m)n^k)$.  After the pre-processing step that is run only
once for the entire matrix, QuickLexSort has a running time linear in the
number of nested sub-matrices to sort. We conclude with an application to
Bayesian network scoring to detect epistasis using SNP marker data.

\end{abstract}

\maketitle
\footnotetext{\today}
\footnotetext{dhaws@us.ibm.com, dchaws@gmail.com}

\section{Introduction}
\label{sec:introduction}
Lexicographical ordering is a method to sort a list of elements where each
element has multiple features, such as a vector, provided one has an order for
each feature. Lexicographic ordering is also known as dictionary or
alphabetical ordering. Put simply, the lexicographical ordering places the
elements in a sequence such that: elements are ordered according to the first
feature, any ties are broken by the second feature, any ties are broken by the
third feature, etc. Lexicographical sorting is a fundamental problem with
applications to contingency tables, Bayesian networks,
databases\cite{poess2003data,lemire2010sorting}, and more. A contingency table
lists the frequency of each element present in the data. For example, given a
matrix, one can form a contingency table which for each unique row, counts the
number of equal rows in the data. A naive approach would loop through each row,
then again loop through the rows and count the number of equal rows. A better
approach would be to sort all the rows of the matrix lexicographically, then
loop through the matrix one last time forming the counts for the contingency
table. Contingency tables formed from a matrix are used in the learning of
Bayesian networks, as well as other applications.

Traditionally, one sorts the rows of a matrix lexicographically by iteratively
applying a stable sort -- a sorting algorithm which preserves the original
order of elements that are equal. The rows of the matrix $D$ are stable sorted by
the least important feature, the next to least important feature, etc. If $D
\in \R^{m \times n}$ and the stable sorting algorithm runs in time $T(m)$, then
the time to sort the matrix is $O(n T(m))$. If the stable sort is a comparison
base sort then $T(m)$ is bounded below by
$\Omega(m\log(m))$\cite{Cormen:2001fk}. In many applications, such as learning
Bayesian networks, one not only wants to sort the rows of a data matrix $D$,
but also sort the rows of $D$ restricted to a sequence of columns.

We present \emph{QuickLexSort} which can efficiently sort a database (rows of a
matrix) restricted to any sequence of features (columns of a matrix). Moreover,
QuickLexSort is designed to quickly sort a nested set of restrictions by
features. For example, one may wish to lexicographically sort all possible
sub-matrices given by all non-empty sets of columns up to a specific
cardinality. QuickLexSort first requires a modest pre-processing step where
each feature of the data set is sorted independently.  When lexicographically
sorting the rows of a single matrix, QuickLexSort (including pre-processing)
has comparable running time to using a stable sort based approach.  However,
when sorting a set of nested sets of features we show QuickLexSort gives a
performance improvement of a multiple of a log factor of the database length
(rows in matrix) over using a standard stable sort based approach. That is,
after the pre-processing step, QuickLexSort has a running time linear in the
number of nested sub-matrices to sort. The pre-processing step need only be
computed once for the entire matrix and not for each sub-matrix. 

The article is organized as follows. In \autoref{sec:background} we
give background details on lexicographical sorting and the stable sorting
approach. In \autoref{sec:qls} we present QuickLexSort and prove the validity and space and
running time. In \autoref{sec:sortingsubmat} we show how
QuickLexSort can be used to efficiently sort nested restrictions of databases.
In \autoref{sec:experiments} we present a small experiment verifying
computationally the advantage of QuickLexSort over a stable sort based approach
to lexicographically sorting a nested set of data base restrictions. In
\autoref{sec:adtrees} we briefly compare QuickLexSort with AD-trees. In
\autoref{sec:bn} we give details on applications of QuickLexSort to scoring and
learning Bayesian networks.  Moreover, we perform experiments showing the
validity to using QuickLexSort and a Bayesian scoring approach to detecting
epistasis in biological SNP data.  Finally, in the Appendix in \autoref{sec:appendix} we
present an augmented version of QuickLexSort which provides a more natural
encoding of the lexicographical ordering.


\section{Background}
\label{sec:background}
Sorting is the method of rearranging a sequence of items such that they are
placed with respect to some order. Here we consider total or linear orderings.
A set follows a \emph{total ordering}, given by the symbol '$\leq$', if the following 
three conditions hold:
\begin{enumerate}
    \item If $a \leq b$ and $b \leq a$ then $ a = b$, (antisymmetry)
    \item If $a \leq b$ and $b \leq c$ then $a \leq c$, (transitivity)
    \item $a \leq b$ or $b \leq a$ (totality).
\end{enumerate}
Here we primarily focus on comparison sorting, a class of sorting
algorithms which only uses a binary comparison operation. That is, the only
information used to sort is the given comparison operation. It is known that the
optimal running time of comparison sort is bounded below by $\Omega(n log(n))$
\cite{Cormen:2001fk}. Some non-comparison sorting algorithms, such as bucket
sort, may run in linear or near linear time, depending on the data.
The exposition here will focus on real valued data and the standard ordering.
However, the results extend naturally to any data with some type of comparison
operation.

Let $D \in \R^{m \times n}$ be a real matrix where we index rows and columns by
$\{0,\ldots,m-1\}$ and $\{0,\ldots,n-1\}$ respectively. 

\begin{definition}[\bf Lexicographic Order]
We say $\ve v \in \R^n$ is less
than $\ve w \in \R^n$ \emph{lexicographically}, denoted $\ve v <_{\text{lex}}
\ve w$, if 
\begin{enumerate}
    \item $v_0 < w_0$ or 
    \item $v_0 = w_0,\ldots, v_k = w_k$ and $v_{k+1} < w_{k+1}$ for some $ 0 \leq k < n-1$. 
\end{enumerate}
\end{definition}
For example $[2,1,3,0] <_{\text{lex}} [2,1,5,-1]$.  If $n=1$, lexicographical
order is equivalent to the normal ordering.  The task here is to sort
the rows of sub-matrices of $D$ (obtained by taking subsets of the columns and
all rows of $D$) using lexicographical order. Here we point out the distinction
of a \emph{set} and a \emph{sequence}, both are groups of objects where in the
former the ordering of the object is irrelevant and in the latter the ordering
matters. That is $(1,3,2)$ and $(3,2,1)$ are distinguishable as sequences but
not as sets.  Typically we will use ``\{\}'' to denote sets and ``()'' to
denote sequences.

By lexicographical sorting of $E \in \R^{m \times n}$ we mean the
lexicographical sorting of the rows (vectors) of $E$ assuming the original
ordering of columns of $E$. By lexicographical sorting of $E$ restricted to the
set of columns $C$ we mean the lexicographical sorting of the sub-matrix of $E$
given by the set of columns $C$ where the original order of the columns is preserved.
By lexicographical sorting of $E$ restricted to the sequence of columns $S$ we
mean the lexicographical sorting of the sub-matrix of $E$ given by the sequence
of columns $S$ where the order of the columns is take from $S$.

\begin{example}
Consider a matrix $E \in \Z^{5 \times 3}$, the lexicographical sorting ($E'$) of $E$,
the lexicographical sorting ($E''$) of the sub-matrix of $E$ given by the set
of columns $\{1,2\}$, and the lexicographical sorting ($E'''$) of the
sub-matrix of $E$ given by the sequence of columns $(2,1)$. In the following,
we write the row and column indices of the original matrix $E$ on the 
left and top of the matrix respectively.
$$E = 
\bordermatrix{
  & 0 & 1 & 2 \cr
0 & 0 & 1 & 0 \cr
1 & 1 & 1 & 0 \cr
2 & 1 & 0 & 0 \cr
3 & 0 & 1 & 1 \cr
4 & 0 & 0 & 1}, \quad
E' =
\bordermatrix{
  & 0 & 1 & 2 \cr
4 & 0 & 0 & 1 \cr
0 & 0 & 1 & 0 \cr
3 & 0 & 1 & 1 \cr
2 & 1 & 0 & 0 \cr
1 & 1 & 1 & 0}, \quad
E'' =
\bordermatrix{
  & 1 & 2 \cr
2 & 0 & 0 \cr
4 & 0 & 1 \cr
0 & 1 & 0 \cr
1 & 1 & 0 \cr
3 & 1 & 1 }, \quad
E''' = 
\bordermatrix{
  & 2 & 1 \cr
2 & 0 & 0 \cr
0 & 0 & 1 \cr
1 & 0 & 1 \cr
4 & 1 & 0 \cr 
3 & 1 & 1}
    $$
\end{example}

We note that all algorithms presented here do not modify the original
data matrices, and simply represent the ordering via certain vectors
which we define below. Note though, that in the examples we do reorder
the rows to better illustrate certain concepts, although we are careful
to preserve the original matrix row indices on the left.

\begin{definition}[\bf Ranking Vector]
Let $D' \in \R^{m \times l}$ be a sub-matrix of $D \in \R^{m \times n}$. 
The unique \emph{ranking vector} $L \in \R^m$ of some 
ordering $\widehat O$ of rows of $D'$ is a vector such that
\begin{enumerate}
    \item $L \in \{0,\ldots,m-1\}^m$,
    \item if row $i$ of $D'$ is equal to row $j$ in the ordering $\widehat O$, then $L_i = L_j$,
    \item if row $i$ of $D'$ is less than row $j$ in the ordering $\widehat O$, then $L_i < L_j$,
    \item $\sum_{i=0}^{m-1} L_i$ is minimal.
\end{enumerate}
\end{definition}
The last item guarantees the ranking vector $L$ is unique. We say ranking
vector $L'$ is a \emph{refinement} of ranking vector $L$ if $L_i \leq L_j$
implies $L'_i \leq L_j$ for all $i,j$.
Intuitively, the $m$ rows of $D$ are sorted lexicographically and thus form
$l$ blocks, where $l \leq m$, every row in a block is lexicographically equal,
and the $l$ blocks are in increasing lexicographic order. In this sense, $L_i$
is the block index in which row $i$ resides.

In what follows we consider a ranking vector sufficient information to describe
an ordering. However, it may be desirable to have an alternative data structure
to describe the ordering, such as an ordered list of row indices giving the
smallest to largest row vectors. We describe such a case and the appropriate 
modifications to our algorithms in the Appendix in \autoref{sec:appendix}.
Our modified algorithm has the same time and space complexity.



\begin{example}
$$D = 
\bordermatrix{
  & 0 & 1 & 2 & 3 & 4 \cr
0 & 1 & 1 & 1 & 2 & 1 \cr
1 & 1 & 1 & 1 & 2 & 0 \cr
2 & 0 & 0 & 0 & 3 & 0 \cr
3 & 1 & 1 & 1 & 2 & 0 \cr
4 & 1 & 0 & 1 & 1 & 1 \cr
5 & 1 & 1 & 1 & 1 & 1 \cr
6 & 1 & 1 & 1 & 3 & 0 \cr
7 & 1 & 1 & 0 & 2 & 1 \cr
8 & 1 & 0 & 1 & 1 & 0 \cr
9 & 1 & 1 & 1 & 1 & 1 }, \;
D' = 
\bordermatrix{
  & 0 & 1 & 2 & 3 & 4 \cr
2 & 0 & 0 & 0 & 3 & 0 \cr
8 & 1 & 0 & 1 & 1 & 0 \cr
4 & 1 & 0 & 1 & 1 & 1 \cr
7 & 1 & 1 & 0 & 2 & 1 \cr
5 & 1 & 1 & 1 & 1 & 1 \cr
9 & 1 & 1 & 1 & 1 & 1 \cr
1 & 1 & 1 & 1 & 2 & 0 \cr
3 & 1 & 1 & 1 & 2 & 0 \cr
0 & 1 & 1 & 1 & 2 & 1 \cr
6 & 1 & 1 & 1 & 3 & 0 }, \;
D'' = 
\bordermatrix{
  & 0 & 3 & 4 \cr
2 & 0 & 3 & 0 \cr
8 & 1 & 1 & 0 \cr
4 & 1 & 1 & 1 \cr
5 & 1 & 1 & 1 \cr
9 & 1 & 1 & 1 \cr
1 & 1 & 2 & 0 \cr
3 & 1 & 2 & 0 \cr
0 & 1 & 2 & 1 \cr
7 & 1 & 2 & 1 \cr
6 & 1 & 3 & 0 }, \;
D''' =
\bordermatrix{
  & 0 & 3 \cr
2 & 0 & 3 \cr
4 & 1 & 1 \cr
5 & 1 & 1 \cr
8 & 1 & 1 \cr
9 & 1 & 1 \cr
0 & 1 & 2 \cr
1 & 1 & 2 \cr
3 & 1 & 2 \cr
7 & 1 & 2 \cr
6 & 1 & 3 }.
$$
A matrix $D$ with row indices $\{0,\ldots,9\}$ and column indices
$\{0,\ldots,4\}$. The matrix $D'$ gives the rows of $D$ sorted
lexicographically, $D''$ gives the rows of $D$ restricted to columns
$\{0,3,4\}$ sorted lexicographically, and $D'''$ gives the rows of $D$
restricted to columns $\{0,3\}$ sorted lexicographically.  The ranking vectors
of the lexicographic orderings shown in $D'$, $D''$, and $D'''$ are $$ L' :=
[6,5,0,5,2,4,7,3,1,4]^\top, \; L'' := [4,3,0,3,2,2,5,4,1,2]^\top, \; L''' :=
[2,2,0,2,1,1,3,2,1,1]^\top.$$ Note the ranking vectors refer to the original
row indices of the matrix $D$.  Note that $L''$ is a refinement of $L'''$.
\label{ex:D}
\end{example}
    
\subsection{Stable Sort}
\begin{definition}[\bf Stable Sort]
A sorting algorithm is \emph{stable} if it maintains the relative order of
items with equal value. That is, if $a$ comes before $b$ in the original input
and $a = b$, then a stable sorting algorithm orders $a$ before $b$.
\end{definition}

\begin{example}
Suppose we performed a stable sorting of the rows of $D$ where we use only the values in
column $4$ to perform the sort.  We preserve the order of all rows which have
the same value in column $4$, and get $D''''$ below.
$$D = 
\bordermatrix{
  & 0 & 1 & 2 & 3 & 4 \cr
0 & 1 & 1 & 1 & 2 & 1 \cr
1 & 1 & 1 & 1 & 2 & 0 \cr
2 & 0 & 0 & 0 & 3 & 0 \cr
3 & 1 & 1 & 1 & 2 & 0 \cr
4 & 1 & 0 & 1 & 1 & 1 \cr
5 & 1 & 1 & 1 & 1 & 1 \cr
6 & 1 & 1 & 1 & 3 & 0 \cr
7 & 1 & 1 & 0 & 2 & 1 \cr
8 & 1 & 0 & 1 & 1 & 0 \cr
9 & 1 & 1 & 1 & 1 & 1 }, \;
D'''' = 
\bordermatrix{
  & 0 & 1 & 2 & 3 & 4 \cr
1 & 1 & 1 & 1 & 2 & 0 \cr
2 & 0 & 0 & 0 & 3 & 0 \cr
3 & 1 & 1 & 1 & 2 & 0 \cr
6 & 1 & 1 & 1 & 3 & 0 \cr
8 & 1 & 0 & 1 & 1 & 0 \cr
0 & 1 & 1 & 1 & 2 & 1 \cr
4 & 1 & 0 & 1 & 1 & 1 \cr
5 & 1 & 1 & 1 & 1 & 1 \cr
7 & 1 & 1 & 0 & 2 & 1 \cr
9 & 1 & 1 & 1 & 1 & 1 }. $$
Although rows $1$ and $2$ have repeated values in column $4$, a stable sorting
algorithm places row $1$ before row $2$, preserving the original
ordering.
\end{example}

A stable sorting algorithm can be used iteratively to perform
lexicographical sorting. When a stable sorting algorithm is used to do
lexicographical sorting we will refer to it as \emph{StableLexSort}.
See \autoref{alg:stablelex} below.

\begin{algorithm}
\begin{algorithmic}[1]
    \REQUIRE $D \in \R^{m\times n}$, $(a_1,\ldots,a_p)$ where $ a_i \in \{0,\ldots,n-1\}$ $\forall i$, $S$ a stable sorting algorithm.
    \ENSURE $D'$ a lexicographic sorting of rows of matrix $D$.
    \STATE Let $D' := D$.
    \FOR{$j= p,\ldots,1$}
        \STATE Sort rows of $D'$ using the stable sorting algorithm $S$ and values in column $a_j$.
    \ENDFOR
    \RETURN $D'$
\end{algorithmic}
\caption{StableLexSort: Lexicographic sort using stable sort.}
\label{alg:stablelex}
\end{algorithm}

\begin{example}
We give an example of \autoref{alg:stablelex} with input $E$ and $(0,1,2,3,4)$.
First stable sort the rows $E$ by the values in column $4$. Further stable sort
the rows by values in column $3$.  Repeat stable sort of the rows by values in
column $2$, then $1$, and finally
$0$.
\begin{align*}
    E := &
\bordermatrix{
  & 0 & 1 & 2 & 3 & 4 \cr
0 & 1 & 1 & 1 & 2 & 1 \cr
1 & 1 & 1 & 1 & 2 & 0 \cr
2 & 0 & 0 & 0 & 3 & 0 \cr
3 & 1 & 1 & 1 & 2 & 0 \cr
4 & 1 & 0 & 1 & 1 & 1 \cr
5 & 1 & 1 & 1 & 1 & 1 \cr
6 & 1 & 1 & 1 & 3 & 0 \cr
7 & 1 & 1 & 0 & 2 & 1 \cr
8 & 1 & 0 & 1 & 1 & 0 \cr
9 & 1 & 1 & 1 & 1 & 1 }, \quad
\bordermatrix{
  & 0 & 1 & 2 & 3 & {\color{red} 4} \cr
1 & 1 & 1 & 1 & 2 & {\color{red} 0} \cr
2 & 0 & 0 & 0 & 3 & {\color{red} 0} \cr
3 & 1 & 1 & 1 & 2 & {\color{red} 0} \cr
6 & 1 & 1 & 1 & 3 & {\color{red} 0} \cr  
8 & 1 & 0 & 1 & 1 & {\color{red} 0} \cr
0 & 1 & 1 & 1 & 2 & {\color{red} 1} \cr
4 & 1 & 0 & 1 & 1 & {\color{red} 1} \cr
5 & 1 & 1 & 1 & 1 & {\color{red} 1} \cr
7 & 1 & 1 & 0 & 2 & {\color{red} 1} \cr
9 & 1 & 1 & 1 & 1 & {\color{red} 1}}, \quad
\bordermatrix{
  & 0 & 1 & 2 & {\color{red} 3} & 4 \cr
8 & 1 & 0 & 1 & {\color{red} 1}  & 0 \cr
4 & 1 & 0 & 1 & {\color{red} 1}  & 1 \cr
5 & 1 & 1 & 1 & {\color{red} 1}  & 1 \cr
9 & 1 & 1 & 1 & {\color{red} 1}  & 1 \cr
1 & 1 & 1 & 1 & {\color{red} 2}  & 0 \cr
3 & 1 & 1 & 1 & {\color{red} 2}  & 0 \cr
0 & 1 & 1 & 1 & {\color{red} 2}  & 1 \cr
7 & 1 & 1 & 0 & {\color{red} 2}  & 1 \cr
2 & 0 & 0 & 0 & {\color{red} 3}  & 0 \cr
6 & 1 & 1 & 1 & {\color{red} 3}  & 0}, \\
  & \phantom{A}\\
  &
\bordermatrix{
  & 0 & 1 & {\color{red} 2} & 3 & 4 \cr
7 & 1 & 1 & {\color{red} 0} & 2 & 1 \cr
2 & 0 & 0 & {\color{red} 0} & 3 & 0 \cr
8 & 1 & 0 & {\color{red} 1} & 1 & 0 \cr
4 & 1 & 0 & {\color{red} 1} & 1 & 1 \cr
5 & 1 & 1 & {\color{red} 1} & 1 & 1 \cr
9 & 1 & 1 & {\color{red} 1} & 1 & 1 \cr
1 & 1 & 1 & {\color{red} 1} & 2 & 0 \cr
3 & 1 & 1 & {\color{red} 1} & 2 & 0 \cr
0 & 1 & 1 & {\color{red} 1} & 2 & 1 \cr
6 & 1 & 1 & {\color{red} 1} & 3 & 0}, \quad  
\bordermatrix{
  & 0 & {\color{red} 1} & 2 & 3 & 4 \cr
2 & 0 & {\color{red} 0} & 0 & 3 & 0 \cr
8 & 1 & {\color{red} 0} & 1 & 1 & 0 \cr
4 & 1 & {\color{red} 0} & 1 & 1 & 1 \cr
7 & 1 & {\color{red} 1} & 0 & 2 & 1 \cr
5 & 1 & {\color{red} 1} & 1 & 1 & 1 \cr
9 & 1 & {\color{red} 1} & 1 & 1 & 1 \cr
1 & 1 & {\color{red} 1} & 1 & 2 & 0 \cr
3 & 1 & {\color{red} 1} & 1 & 2 & 0 \cr
0 & 1 & {\color{red} 1} & 1 & 2 & 1 \cr
6 & 1 & {\color{red} 1} & 1 & 3 & 0}, \quad  
\bordermatrix{
  & {\color{red} 0} & 1 & 2 & 3 & 4 \cr
2 & {\color{red} 0} & 0 & 0 & 3 & 0 \cr
8 & {\color{red} 1} & 0 & 1 & 1 & 0 \cr
4 & {\color{red} 1} & 0 & 1 & 1 & 1 \cr
7 & {\color{red} 1} & 1 & 0 & 2 & 1 \cr
5 & {\color{red} 1} & 1 & 1 & 1 & 1 \cr
9 & {\color{red} 1} & 1 & 1 & 1 & 1 \cr
1 & {\color{red} 1} & 1 & 1 & 2 & 0 \cr
3 & {\color{red} 1} & 1 & 1 & 2 & 0 \cr
0 & {\color{red} 1} & 1 & 1 & 2 & 1 \cr
6 & {\color{red} 1} & 1 & 1 & 3 & 0}, \quad  
\end{align*}
\end{example}

\begin{proposition}
If $D \in \R^{m \times n}$, and the running time of the stable sort algorithm
$S$ is $T(m)$, then \autoref{alg:stablelex} sorts in time $O(nT(m))$.
\end{proposition}

Thus if the stable sort algorithm $S$ is a comparison sort, then the running
time of \autoref{alg:stablelex} is bounded below by $\Omega(nm\log(m))$.
For example, if Merge sort \cite{Knuth:1998fk} was used, which has running time
of $O(m\log(m))$, then the running time of StableLexSort on $D$ would be
$O(nm\log(m))$.

\section{QuickLexSort}
\label{sec:qls}
Here we present a new algorithm for lexicographical sorting called
\emph{QuickLexSort}.  We will show that the running time of QuickLexSort is
comparable to StableLexSort when sorting a single matrix. Moreover, we will
demonstrate that QuickLexSort is considerably faster than StableLexSort when
performing multiple lexicographic sorts of related sub-matrices.

The proposed algorithm QuickLexSort first requires each column of $D$ to be
independently sorted and stored. The results are stored in the matrix $Q \in \Z_+^{m \times
n}$ where the $j$th column $Q_{\cdot j}$ of $Q$ stores the row indices
$\{0,\ldots,m\}$ after sorting the $j$th column $D_{\cdot j}$ of $D$.  

\begin{example}
$$D = 
\bordermatrix{
  & 0 & 1 & 2 & 3 & 4 \cr
0 & 1 & 1 & 1 & 2 & 1 \cr
1 & 1 & 1 & 1 & 2 & 0 \cr
2 & 0 & 0 & 0 & 3 & 0 \cr
3 & 1 & 1 & 1 & 2 & 0 \cr
4 & 1 & 0 & 1 & 1 & 1 \cr
5 & 1 & 1 & 1 & 1 & 1 \cr
6 & 1 & 1 & 1 & 3 & 0 \cr
7 & 1 & 1 & 0 & 2 & 1 \cr
8 & 1 & 0 & 1 & 1 & 0 \cr
9 & 1 & 1 & 1 & 1 & 1 }, \;
Q = 
\bordermatrix{
  & 0 & 1 & 2 & 3 & 4 \cr
  & 2 & 4 & 2 & 4 & 3 \cr
  & 4 & 2 & 7 & 5 & 8 \cr
  & 5 & 8 & 4 & 9 & 6 \cr
  & 6 & 5 & 6 & 8 & 2 \cr
  & 0 & 6 & 0 & 0 & 1 \cr
  & 7 & 7 & 1 & 7 & 4 \cr
  & 9 & 9 & 9 & 1 & 7 \cr
  & 8 & 1 & 8 & 3 & 9 \cr
  & 3 & 3 & 3 & 2 & 0 \cr
  & 1 & 0 & 5 & 6 & 5}. $$
A matrix $D$ and the matrix $Q$ storing the sort of the columns of $D$
described above. E.g., reading down the $0$th column of Q, for column $0$ of $D$, the smallest
entry is in row $2$, followed by row $4$, followed by row $5$, etc.
\end{example}

The QuickLexSort algorithm sorts (conceptually) by iteratively appending 
columns to the current matrix and sorting, until the desired sequence of columns is reached. That is,
\autoref{alg:qlsrefine} refines the current sort with respect to the sequence of columns
$(a_1,\ldots,a_j)$ to give a sort with respect to the sequence of columns
$(a_1,\ldots,a_{j+1})$. In some sense this is opposite of StableLexSort.  In
StableLexSort one stable sorts from the least important column to the most
important. In QuickLexSort one sorts from the most important column to the
least important, refining the ranking vector as it goes.

\autoref{alg:qlsrefine} (QuickLexSortRefine), is the core of the methods described here.
\autoref{alg:qlsrefine} takes as input the matrix $D$ [$D$], the sorting of the
columns of $D$ [$Q$], the column to refine $L$ by [$i$], and the
current ranking vector [$L$]. It returns the refined ranking vector $L'$. That
is, if the input ranking vector $L$ represents the sorting of the rows of $D$
(restricted to some sequence of columns), the returned ranking vector $L'$
represents the refined sorting where we consider appending the $ith$ column of
$D$. Again, if the input ranking vector $L$ represents the lexicographical
sorting of a matrix
\begin{equation*}
    D' = 
\bordermatrix{
  & j_0 & \cdots & j_p \cr
0 &   D_{0j_0} & \cdots & D_{0j_p} \cr
1 &   D_{1j_0} & \cdots & D_{1j_p} \cr
\vdots  & \vdots & \ddots & \vdots \cr
n &   D_{nj_0} & \cdots & D_{nj_p} }
\end{equation*}
the ranking vector $L'$ output from \autoref{alg:qlsrefine} represents the lexicographical sorting of the matrix
\begin{equation*}
D'' = 
\bordermatrix{
    & j_0 & \cdots & j_p & i\cr
  0 &   D_{0j_0} & \cdots & D_{0j_p} & D_{0i} \cr
1 &   D_{1j_0} & \cdots & D_{1j_p} & D_{1i}  \cr
\vdots  & \vdots & \ddots & \vdots \cr
n &   D_{nj_0} & \cdots & D_{nj_p} & D_{ni}  }.
\end{equation*}

\begin{algorithm}[!h]
\begin{algorithmic}[1]
    \REQUIRE $D \in \R^{m\times n}$, $Q \in \Z^{m\times n}$, $i \in \{0,\ldots,n-1\}$, $L \in \Z^m$.
    \ENSURE $L' \in Z^m$.
    \STATE $L' := \ve 0 \in Z^m$. 
    \STATE $\IDval := \ve 0 \in Z^m$. \# Records most recent value in $D$ w.r.t.\ ID.
    \STATE $\IDvalInit := \ve 0 \in Z^m$.
    \STATE $subID := \ve 0 \in Z^m$. \# Records subID of $L$.
    \STATE $newCount := \ve 0 \in Z^m$. \# Records count of refinements of each input ID of L.
    \FOR{ $j = 0,\ldots,m-1$ }
        \IF{ $\IDvalInit[L[Q[j,i]]] == 0$} 
            \STATE $\IDvalInit[L[Q[j,i]]] := 1$.
            \STATE $\IDval[L[Q[j,i]]] := D[Q[j,i],i]$.
        \ELSE
            \IF{ $\IDval[L[Q[j,i]]]\ != D[Q[j,i],i]$}
                \STATE $\IDval[L[Q[j,i]]] := D[Q[j,i],i]$.
                \STATE $newCount[L[Q[j,i]] := newCount[L[Q[j,i]] + 1$.
            \ENDIF
        \ENDIF
        \STATE $subID[Q[j,i]]] := newCount[L[Q[j,i]]]$.
    \ENDFOR
    \STATE $numNewID := 0 \in \Z^M$.
    \STATE $numNewID[m-1] := \sum_{j=0}^{m-2} newCount[j]$.
    \FOR{ $j = m-2,\ldots,1$ }
        \STATE $numNewID[j] := numNewID[j+1] - newCount[j]$.
    \ENDFOR
    \FOR{ $j = 0,\ldots, m-1$}
        \STATE $L'[j] := L[j] + numNewID[L[j]] + subID[j]$.
    \ENDFOR
    \RETURN $(L')$.
\end{algorithmic}
\caption{QuickLexSortRefine}
\label{alg:qlsrefine}
\end{algorithm}

Intuitively, the task of \autoref{alg:qlsrefine} is to 1) preserve the current
ordering, i.e.\ if row $j$ was lexicographically smaller than row $k$, then
this is true in the new order, 2) all previous rows that were lexicographically
equal should be sorted given the newly appended column $i$. The novelty of
\autoref{alg:qlsrefine} is that it performs the second item above in linear
time using the pre-computed ordering of the newly appended column $i$. 

\begin{example}
Consider matrix $D$ in \autoref{ex:D}. Suppose $L$ is the ranking vector of the
lexicographical ordering of the sub-matrix given by columns $\{0,3\}$ of $D$
and we then perform \autoref{alg:qlsrefine} with $i=4$. Thus, part of the input would be 
\begin{equation*}
\bordermatrix{
  & L\cr 
0 & 2\cr
1 & 2\cr
2 & 0\cr
3 & 2\cr
4 & 1\cr
5 & 1\cr
6 & 3\cr
7 & 2\cr
8 & 1\cr
9 & 1}, \quad
\bordermatrix{
  & D_i\cr
0 & 1\cr
1 & 0\cr
2 & 0\cr
3 & 0\cr
4 & 1\cr
5 & 1\cr
6 & 0\cr
7 & 1\cr
8 & 0\cr
9 & 1}, \quad
\bordermatrix{
  & Q_i\cr
0 & 3\cr
1 & 8\cr
2 & 6\cr
3 & 2\cr
4 & 1\cr
5 & 4\cr
6 & 7\cr
7 & 9\cr
8 & 0\cr
9 & 5}.
\end{equation*}

The progression of the vectors $IDval$, $subID$, and $newCount$ are shown from
left to right as the for loop on line 6 goes from $j=0$ to $j=m-1$, Note, ``$\cdot$''
signifies unassigned values.
\begin{align*}
    IDval &= 
\bordermatrix{
j=  &  0      &  1      &  2      &  3  &  4  &  5  &  6  &  7  &  8  &  9  \cr
0   &  \cdot  &  \cdot  &  \cdot  &  0  &  0  &  0  &  0  &  0  &  0  &  0  \cr
1   &  \cdot  &  0      &  0      &  0  &  0  &  1  &  1  &  1  &  1  &  1  \cr
2   &  0      &  0      &  0      &  0  &  0  &  0  &  1  &  1  &  1  &  1  \cr
3   &  \cdot  &  \cdot  &  0      &  0  &  0  &  0  &  0  &  0  &  0  &  0  }\\
\phantom{A}\\
newCount &= 
\bordermatrix{
j=  &  0  &  1  &  2  &  3  &  4  &  5  &  6  &  7  &  8  &  9  \cr
0   &  0  &  0  &  0  &  0  &  0  &  0  &  0  &  0  &  0  &  0  \cr
1   &  0  &  0  &  0  &  0  &  0  &  1  &  1  &  1  &  1  &  1  \cr
2   &  0  &  0  &  0  &  0  &  0  &  0  &  1  &  1  &  1  &  1  \cr
3   &  0  &  0  &  0  &  0  &  0  &  0  &  0  &  0  &  0  &  0  }\\
\phantom{A}\\
subID &= 
\bordermatrix{
j=  &  0      &  1      &  2      &  3      &  4      &  5      &  6      &  7      &  8      &  9  \cr
0   &  \cdot  &  \cdot  &  \cdot  &  \cdot  &  \cdot  &  \cdot  &  \cdot  &  \cdot  &  1      &  1  \cr
1   &  \cdot  &  \cdot  &  \cdot  &  \cdot  &  0      &  0      &  0      &  0      &  0      &  0  \cr
2   &  \cdot  &  \cdot  &  \cdot  &  0      &  0      &  0      &  0      &  0      &  0      &  0  \cr
3   &  0      &  0      &  0      &  0      &  0      &  0      &  0      &  0      &  0      &  0  \cr
4   &  \cdot  &  \cdot  &  \cdot  &  \cdot  &  \cdot  &  1      &  1      &  1      &  1      &  1  \cr
5   &  \cdot  &  \cdot  &  \cdot  &  \cdot  &  \cdot  &  \cdot  &  \cdot  &  \cdot  &  \cdot  &  1  \cr
6   &  \cdot  &  \cdot  &  0      &  0      &  0      &  0      &  0      &  0      &  0      &  0  \cr
7   &  \cdot  &  \cdot  &  \cdot  &  \cdot  &  \cdot  &  \cdot  &  1      &  1      &  1      &  1  \cr
8   &  \cdot  &  0      &  0      &  0      &  0      &  0      &  0      &  0      &  0      &  0  \cr
9   &  \cdot  &  \cdot  &  \cdot  &  \cdot  &  \cdot  &  \cdot  &  \cdot  &  1      &  1      &  1  }
\end{align*}

\end{example}

Next we present \autoref{alg:qls} (QuickLexSort) which, we will prove,
lexicographically sorts a sub-matrix restricted to a sequence of columns of
$D$, requiring $Q$ an initial sorting of the columns of $D$.

\begin{algorithm}[!h]
\begin{algorithmic}[1]
    \REQUIRE $D \in \R^{m\times n}$, $Q \in \Z^{m\times n}$, $(a_1,\ldots,a_p)$ where $ a_i \in \{0,\ldots,n-1\}$ $\forall i$.
    \ENSURE $L' \in Z^m$.
    \STATE $L' := \ve 0 \in \R^n$.
    \FORALL{ $i \in \{a_1,\ldots,a_p\}$}
        \STATE $L' := QuickLexSortRefine(D,Q,i,L')$.
    \ENDFOR
    \RETURN $L'$.
\end{algorithmic}
\caption{QuickLexSort}
\label{alg:qls}
\end{algorithm}

We now prove the validity and running times of \autoref{alg:qlsrefine}
and \autoref{alg:qls}. 

\begin{lemma} 
\label{lem:qlsrefine}
If $D \in \R^{m \times n}$, $Q \in \R^{m \times n}$ where the $j$th column of
$Q$ stores the sorting of the $j$th column of
$D$, $L$ is the ranking vector of the lexicographical sorting of the
sub-matrix of $D$ determined the sequence of columns
$(a_1,\ldots,a_p)$, and $i \in \{0,\ldots,n-1\}$, then 
\autoref{alg:qlsrefine} returns the ranking vector $L'$ of the lexicographical
sorting of the sub-matrix of $D$ determined by the sequence of columns
$(a_1,\ldots,a_p,i)$.
\end{lemma}
\begin{proof}
We need to prove 
\begin{enumerate}
    \item if $L[j]< L[k]$ then $L'[j]< L'[k]$, and
    \item when $L[j]= L[k]$ we have $D_{ji} < D_{ki}$ if and only if $L'[j] < L'[k]$.
\end{enumerate}
First, we note the for
loop on line 6 visits the elements of $D_{\cdot i}$ in increasing order by using the
data structure $Q$. Thus, for all ranks $r$ in L, all elements of $D_i$ of rank
$r$ will be visited in increasing order. This is the crux of the validity of
\autoref{alg:qlsrefine} and is worthwhile to repeat. Given the current
lexicographical order given by $L$, every row of $D$, restricted to the
sequence of columns $(a_1,\ldots,a_p)$ has some rank $r$. Because \autoref{alg:qlsrefine}
uses $Q_{\cdot i}$, the algorithm will visit all the rows of current rank $r$ 
in the order given by the new column $i$ of interest.

We first claim that $newCount[r]$ is equal to the number of unique elements
of $D_{\cdot i}$ of rank $r$, with respect to $L$. The data structure $IDval[r]$ records
the most recently observed value of $D_{\cdot i}$ of rank $r$. The data structure
$IDvalInit[r]$ simply denotes if nothing has been observed yet. Thus, when
we observe an element of $D_{\cdot i}$ of rank $r$ that differs from $IDval[r]$,
we update $IDval[r]$ (line 9 and line 12) and increase $newCount[r]$ by
one (line 13).

Second, we claim that $subID$ restricted to all elements of rank $r$
is a ranking vector over the elements of $D_{\cdot i}$ restricted to elements of rank $r$.
More specifically, $subID[j]$ is equal to the number of unique entries in $D_{\cdot i}$
of rank $L[j]$ strictly less than $D_{j i}$. The variable $subID[j]$ is initialized 
to be zero and is set to the current value of $newCount[L[j]]$ (line 16). That is,
$subID[j]$ is set to the current number of unique elements of $D_{\cdot i}$ of
rank $L[j]$.

The vector $numNewID$ is simply a partial sum (offset by one index) of the vector
$newCount$. We can now prove the two important properties required to complete
the proof.  Suppose $L[j] < L[k]$ and consider

\begin{align*}
    & L'[k]  -  L'[j] \\ 
    = & L[k] + numNewID[L[k]] + subID[k] -  L[j] - numNewID[L[j]] - subID[j] \\ 
    = &L[k] + \sum_{l=0}^{L[k]-1} newCount[l] + subID[k] - L[j] - \sum_{l=0}^{L[j]-1} newCount[l] - subID[j] 
\end{align*}

and note $subID[j] \leq newCount[L[j]]$.
Thus we have 
\begin{align*}
    = &L[k] + \sum_{l=0}^{L[k]-1} newCount[l] + subID[k] - L[j] - \sum_{l=0}^{L[j]-1} newCount[l] - subID[j] \\
    = &L[k] + \sum_{l=L[j]-1}^{L[k]-1} newCount[l] + subID[k] - L[j]- subID[j] \\
    \geq &L[k] + \sum_{l=L[j]}^{L[k]-1} newCount[l] + subID[k] - L[j] \geq 0,
\end{align*}
and therefore $L'[j] < L'[k]$.

Lastly, if $L[j] = L[k]$ then considering the definition of $L[j]$ and $L[k]$
(line 24) we see the only variable is $subID$. We have already shown
that $subID$ is a ranking vector of items of the same rank. Thus $D_{ji} < D_{ki}$
if and only if $subID[j] < subID[k]$ and the claim is proved.
\end{proof}

\begin{lemma} 
\label{lem:qls}
If $D \in \R^{m \times n}$, $Q \in \R^{m \times n}$ where the $j$th column of
$Q$ stores the sorting of the $j$th column of
$D$, $(a_1,\ldots,a_p)$ where $ a_i \in \{0,\ldots,n-1\}$ $\forall i$, then
\autoref{alg:qls} returns the ranking vector $L'$ of the lexicographical
sorting  of the sub-matrix of $D$ determined by the sequence of columns
$(a_1,\ldots,a_p)$.
\end{lemma}
\begin{proof}
Since \autoref{alg:qlsrefine} refines the ranking vector for each newly
appended column, the result follows.
\end{proof}

\begin{lemma} 
\autoref{alg:qlsrefine} runs in time $O(m)$ and space $O(m)$.
\label{lem:qlsrefineruntime}
\end{lemma}
\begin{proof}
There are only three loops in \autoref{alg:qlsrefine}, each of them
repeated $m$ times. Each inner operation is constant time.  The only space
requirements are determined by column vectors of the $m \times n $ input
matrices and the vectors of length $m$.
\end{proof}

\begin{lemma} 
\autoref{alg:qls} runs in time $O(mp)$ and space $O(mn)$.
\end{lemma}
\begin{proof}
There are $p$ calls made to \autoref{alg:qlsrefine} which by
\autoref{lem:qlsrefineruntime} imply the total running time is $O(mp)$.
The only space requirments are determined by the $m \times n $input matrices
and the vectors of length $m$.
\end{proof}

Recall that both \autoref{alg:qlsrefine} and \autoref{alg:qls} require the
columns of $D$ to be sorted and recorded in the input $Q$. Thus to
lexicographically sort a matrix $D \in \R^{m \times n}$ using
\autoref{alg:qls} requires $O(nm\log(m) + nm) = O(nm\log(m))$, where we use an
$O(m\log(m))$ comparison sort to find $Q$.

\section{Sorting Sub-Matrices}
\label{sec:sortingsubmat}

Consider the problem of sorting all sub-matrices of $D$ given by every possible
sequence of columns. 

\begin{problem}[\bf Sort All Sub-Matrices Given By Column Sequences] $\phantom{a}$ \\
Let $D \in \R^{m \times n}$. 
\begin{itemize}
    \item For every sequence of columns $(a_1,\ldots,a_p)$ where, $1 \leq p \leq n$, $a_i \neq a_j \, \forall i,j$, $a_i \in \{0,\ldots,n-1\}$ $\forall i$:
    \begin{itemize}
        \item Lexicographically sort the sub-matrix of $D$ determined by the sequence of columns $(a_1,\ldots,a_p)$.
    \end{itemize}
    \label{prob:allseq}
\end{itemize}
\end{problem}

Also consider the sub-problem of sorting all sub-matrices of 
$D$ given by every possible subset of columns.

\begin{problem}[\bf Sort All Sub-Matrices Given By Column Sets] $\phantom{a}$ \\
Let $D \in \R^{m \times n}$. 
\begin{itemize}
    \item For every non-empty subset of columns $\{a_1,\ldots,a_p\} \subseteq \{0,\ldots,n-1\}$:
    \begin{itemize}
        \item Lexicographically sort the sub-matrix of $D$ determined by the set of columns $\{a_1,\ldots,a_p\}$.
    \end{itemize}
    \label{prob:allsubset}
\end{itemize}
\end{problem}

In \autoref{prob:allseq} there are $\sum_{i=1}^n \frac{n!}{(n-i)!}$ non-empty
sub-matrices to consider.  In \autoref{prob:allsubset} there are $2^n -1 $
non-empty sub-matrices to consider. Both StableLexSort (\autoref{alg:stablelex})
and QuickLexSort (\autoref{alg:qls}) can be used to solve \autoref{prob:allseq}
and \autoref{prob:allsubset}. One simply enumerates the set of sub-matrices and
applies either algorithm.

We now present how the core of the QuickLexSort Algorithm
(\autoref{alg:qlsrefine}) lends itself ideally to \autoref{prob:allseq} and
\autoref{prob:allsubset}. That is, we can use \autoref{alg:qlsrefine} to
efficiently sort all the nested sub-matrices. The new \autoref{alg:qlsallperms}
for \autoref{prob:allseq} enumerates all sequences of columns in a
depth-first-search (DFS) manner.  It then exploits the fact that
\autoref{alg:qlsrefine} will take a current ranking vector and refine it by
considering appending an additional column. In this way we save the current
ranking vector and refine it based on all possible ways to append a column to
the current sub-matrix. 

\begin{algorithm}
\begin{algorithmic}[1]
    \REQUIRE $D \in \R^{m\times n}$, $Q \in \Z^{m\times n}$, $(a_1,\ldots,a_p)$ where $ a_i \in \{0,\ldots,n-1\}$ $\forall i$.
    \FOR{ $i \in \{0,\ldots,n\} \setminus \{a_1,\ldots,a_p\}$}
        \STATE $L' := $QuickLexSortRefine($D$,$Q$,$i$,$L$).
        \STATE Print $L'$.
        \STATE QuickLexSortAllSeq($D$,$Q$,$(a_1,\ldots,a_p,i)$,$L'$)
    \ENDFOR
\end{algorithmic}
\caption{QuickLexSortAllSeq}
\label{alg:qlsallperms}
\end{algorithm}

\autoref{alg:qlsallperms} is initially called with QuickLexSortAllSeq($D$,$Q$,$()$,$0\in \R^m$).

\begin{lemma}
    \autoref{alg:qlsallperms} has running time $O\left(m \sum_{i=1}^n m \frac{n!}{(n-i)!}\right)$ and space requirements $O(mn)$.
\end{lemma}
\begin{proof}
    Exactly $\frac{n!}{(n-i)!}$ calls are made to \autoref{alg:qlsrefine}, which itself has running time and space $O(nm)$.
\end{proof}

As it stands, StableLexSort could be used  
inside \autoref{alg:qlsallperms} but would not achieve the same running time.
If we replaced QuickLexSort (\autoref{alg:qlsallperms}) with StableLexSort (\autoref{alg:stablelex}) on
line 2 of \autoref{alg:qlsallperms} then the running time would
increase to $O\left(m log(m) \sum_{i=1}^n m \frac{n!}{(n-i)!}\right)$.

This highlights the distinct advantage of QuickLexSort: It is linear time to
refine the lexicographical sorting when appending a column, provided the
columns of the data matrix have been pre-sorted.

Naively one may think to use a stable sort algorithm and append the columns in
the opposite order (since it has to work from least to most important columns),
and proceed in a DFS manner to explore all possible sorting. However, the
stable sort can not take advantage of the information contained in $Q$ and
would still need to do a comparison sort on each new column. 

With minor alteration of \autoref{alg:qlsallperms} we can handle \autoref{prob:allsubset}.

\begin{algorithm}
\begin{algorithmic}[1]
    \REQUIRE $D \in \R^{m\times n}$, $Q \in \Z^{m\times n}$, $\{a_1,\ldots,a_p\}$ where $ a_i \in \{0,\ldots,n-1\}$ $\forall i$.
    \FOR{ $i$ such that $n > i > \max(\{a_1,\ldots,a_p\}$}
        \STATE $L' := $QuickLexSortRefine($D$,$Q$,$i$,$L$).
        \STATE Print $L'$.
        \STATE QuickLexSortAllSubsets($D$,$Q$,$\{a_1,\ldots,a_p,i\}$,$L'$)
    \ENDFOR
\end{algorithmic}
\caption{QuickLexSortAllSubsets}
\label{alg:qlsallsubset}
\end{algorithm}

\autoref{alg:qlsallsubset} is initially called with QuickLexSortAllSubsets($D$,$Q$,$()$,$0\in \R^m$).

\begin{lemma}
    \autoref{alg:qlsallsubset} has running time $O\left(m 2^n\right)$ and space requirements $O(mn)$.
\end{lemma}
\begin{proof}
    Exactly $2^n - 1$ calls are made to \autoref{alg:qlsrefine}, which itself has running time and space $O(nm)$.
\end{proof}

Again, attempting to use StableLexSort on line 2 of \autoref{alg:qlsallsubset}
would increase the running time to $O(m\log(m) 2^n)$. In both cases, this gain
may seem modest given the dominating terms involving $n$. However, we note that
in many applications one may not in fact enumerate all sub-matrices but will
instead enumerate all nested sub-matrices up to a certain cardinality.  For
example if one wishes to enumerate all sub-matrices with up to two columns
then the running time of using QuickLexSort is $O(m n^2)$ compared to $O(m \log(m) n^2)$ 
for StableLexSort.

In general consider a set of nested sub-matrices indexed by their sequence of
columns $\mathcal A$, and let $|\mathcal A|$ denote the size of $A$. Nested in
the sense that if $A \in \mathcal A$ then either $A$ is a singleton or there
exist $B \in \mathcal A$ such that $A$ and $B$ differ by one element. Then if
one can efficiently (linear in $|\mathcal A|$) enumerate the sub-matrices given
by $\mathcal A$ then the running time to sort all $| \mathcal A|$ sub-matries
using QuickLexSort is $O(m | \mathcal A |)$.  Extending the previous example,
if one wishes to sort all sub-matrices with up to $k$ columns, then the
running time of QuickLexSort is $O(m n^k)$.

\section{Experiments}
\label{sec:experiments}
As a verification of the running times of \autoref{alg:qlsrefine} and
\autoref{alg:qls} claimed in \autoref{lem:qlsrefine} and \autoref{lem:qls}, we
performed a short experiment using the {\tt Poker Hand} data set from the
University of California, Irvine's Machine Learning Repository
\cite{Bache+Lichman:2013}. The data set consists of a matrix with $25,010$ rows
and $7$ columns with discrete numerical values. Ten data sets were created for
the experiments, consisting of the first $10\%$, $20\%$, $\ldots$, $100\%$
rows. For each of the ten data sets, QuickLexSort (using merge sort for the
preliminary sorting of data columns) and StableLexSort (using merge sort) were
run to sort all possible non-empty subsets of $7$ columns. Note, the running
times for QuickLexSort includes the pre-sorting step.
\autoref{fig:SLSvsQLS} shows the time to lexicographically sort using both
methods. It is fairly easy to see the linear growth in running time of
QuickLexSort compared to the linear times log factor running time of StableLexSort with
respect to the number of rows.

\begin{figure}
    \begin{center}
    \includegraphics[width=12cm]{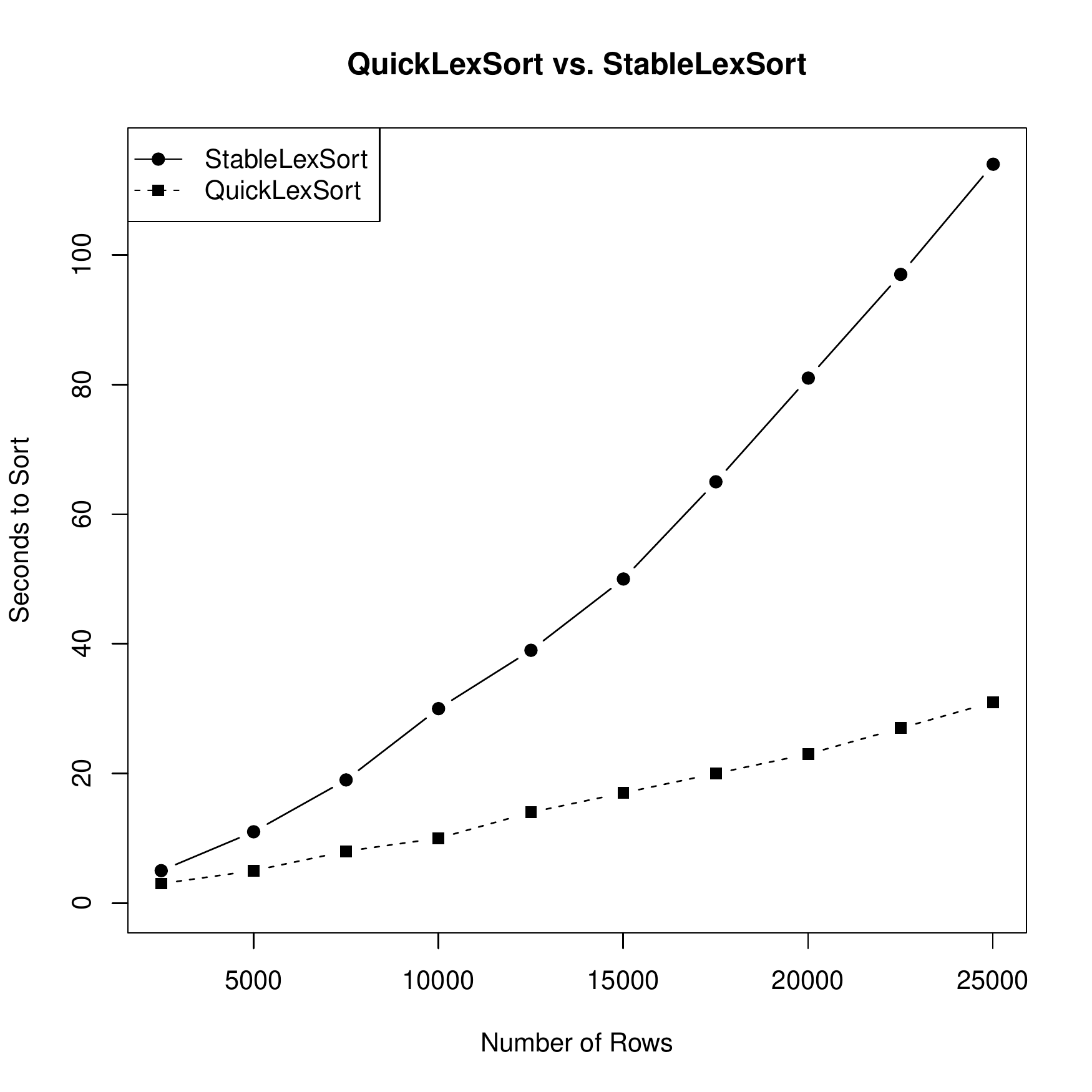}
    \end{center}
    \caption{Running times of QuickLexSort and StableLexSort to sort all
    possible subsets of seven columns on ten truncations of the {\tt Poker
Hand} data set from the UCI Machine Learning data base.}
\label{fig:SLSvsQLS}
\end{figure}

%
%

\section{Comparison to AD-trees}
\label{sec:adtrees}
A popular method which specifically computes contingency tables (and sorts) is
ADtrees \cite{Moore:1998fk}. Although retrieving a contingency table (or
sorting) can be very fast -- faster than QuickLexSort -- the time and space
requirements to compute and store the required data structures can be enormous. Assuming
binary features (features only take two values), the cost to build an AD
tree is bounded above by 
\begin{equation} 
\sum_{k=0}^{\lfloor log_2(m) \rfloor} \frac{m}{2^k} { n \choose k }, 
\label{eq:ad1}
\end{equation} 

where $D \in \R^{m \times n}$. If all possible combinations
of the binary features appear in $D$, the space requirement would be $2^n$.
Even with a reasonable number of rows the space requirement would be bound
above by 

\begin{equation} 
    \sum_{k=0}^{\lfloor log_2(m) \rfloor} { n \choose k }. 
\label{eq:ad2}
\end{equation} 
The time
and space requirements can become practically prohibitive as $n$ and $m$ grow.
For example, constructing and storing the ADtree for a dataset $D \in \R^{1000
\times 50000}$ would be infeasible.  By contrast, QuickLexSort only requires
linear space and time for each sort.

\section{Applications to Bayesian Networks}
\label{sec:bn}
Bayesian networks (BN) are graphical models that have applications in a
plethora of areas including machine learning, statistical inference, finance,
biology, artificial intelligence, etc \cite{kollar2009probabilistic,
studeny2005probabilistic}. Bayesian networks represent the conditional
independences in some given data and are modeled through directed acyclic
graphs (DAGs). In a naive sense, the task of learning the BN structure is to
explore all possible DAGs and choose the DAG which best fits the data. Note
that learning the BN structure is NP-hard \cite{chickering1996learning,chickering2004large}. The fit
of a proposed DAG to the data is evaluated by a scoring function such as
Bayesian information criteria (BIC) or Bayesian Dirichlet equivalence (BDE). 
The BIC and BDE graph scoring functions evaluate a graph $G$ by looking
at each node $i$ and its parents $Pa(i)$ (the nodes which have edges directed
to $i$ in $G$). For every pair $(i,Pa(i))$, BIC and BDE compute local scores (a
score depending only on $(i,Pa(i))$). In the end, the score of the proposed
graph $G$ is a sum over all local scores $(i,Pa(i))$.  At
a low level, the BIC or BDE local score of $(i,Pa(i))$ simply requires two
contingency tables: 1) a contingency table of the input data matrix restricted
to columns indexed by $P(i)$, 2) a contingency table of the input data matrix
restricted to columns indexed by $i \cup Pa(i)$.

In many current methods to learn a BN, the task is roughly broken into two
steps. In step one, all local scores are precomputed. Often this is prohibitive
and in practice only the local scores up to a certain cardinality are computed,
or steps are taken to theoretically exclude certain local scores
\cite{de2011efficient,de2009structure}.  In step two, the structure is
learned by some intelligent method (Integer programming, Dynamic Programming,
Heuristically, and more)\cite{de2011efficient,chickering2003optimal,de2009structure,de2000new,jaakkola2010learning,barlettadvances,cussens2012bayesian,silander2012simple,singh2005finding}. Research has focused mainly on the
second step. However, the first step of local scoring merits exploration. For
example the condition that parent sets are limited in cardinality can be quite
artificial. But, if $m$ and $n$ become large, it may be prohibitive to compute
all local scores using an approach such as StableLexSort. However, QuickLexSort
is fast and requires small space.  Moreover, approaches can be taken in which
the local scores are done on-the-fly, in which case QuickLexSort can be of use.

One approach to learning BN is to perform a heuristic search of the solution
space by iteratively changing the current graph
structure\cite{moore2003optimal,madigan1994model,madigan1995bayesian,giudici1999decomposable}. Again, in many cases it may be
infeasible to store all contingency tables and it would be better to score each
new graph. The proposed new graph can be chosen such that previous contingency
tables can be updated efficiently by QuickLexSort. For example if the graphical
moves are restricted to simply adding or removing a single edge of the current
graph.  As a new approach, we are currently developing a method which explores
the solution space using characteristic imsets\cite{hemmecke2012characteristic,studeny2010characteristic,studeny2011polyhedral,StudenyHaws:2013} -- a more natural encoding of unique probability models forming BNs --
and QuickLexSort in order to efficiently move through the solution space.

A biological example of an application of Bayesian networks is the modeling of epistasis --
the interaction of multiple genes to produce a phenotype. Using Bayesian
networks (and related measures dependent on contingency tables) has
proven useful in detecting epistasis
\cite{sucheston2010comparison,jiang2011learning,shang2011performance}.  In this
case, one does not need to consider the full class of DAGs, and the problem
reduces to simply scoring.  Suppose we are given the genotypes for $1,000$
individuals each with $50,000$ single nucleotide polymorphisms (SNPs) and some
phenotype (disease/no-disease). Then in this case we have a matrix $D \in
\R^{1000 \times 50000}$.  The task of detecting k-way epistasis using Bayesian
networks reduces to computing contingency tables (sorting) all possible choices
of $k$ subsets columns of $D$.  Considering \autoref{eq:ad1} and
\autoref{eq:ad2}, it would be impractical to use ADtrees.  For most $k>2$, it
would be infeasible to store all contingency tables for all $ 50,000 \choose k$
SNP $k$-tuples. However, QuickLexSort requires linear time and space, to check
each choice of k columns of $D$.

As a preliminary experiment, QuickLexSort was used for Bayesian network scoring
and the detection of epistasis on Maize genotype and phenotype data
\cite{rincent2012maximizing} used for the European CornFed program. The data
used consisted of $261$ inbred \emph{dent} maize plant lines (rows) and
$30,027$ SNPs (columns). The phenotype used was male flowering time. The
Bayesian Dirichlet equivalent (BDE) \cite{heckerman1995learning} score was used
to detect up to two-way epistasis. In general, to compute the BDE score of any
$k$ SNPs with respect to the phenotype, one needs two contingency tables: the
contingency table given by the sub-matrix over those $k$ SNPs as well as the
contingency table given by the sub-matrix over the $k$ SNPs with the additional
column of phenotypes. QuickLexSort lends itself naturally to this process
since we iterate through the sets of SNPs in depth-first-search manner. Moreover,
to compute the latter contingency table we simply append the phenotype column
and do one call to \autoref{alg:qlsrefine}. For the experiment, the BDE score was computed
for all singleton and pairs of SNPs. The QuickLexSort based approach took
approximately $4.99$ hours to perform ${30027 \choose 2} - 1 + 30027 =
450,825,379$ BDE scores on the $261 \times 30,027$ data matrix, which required
twice as many contingency table computations.



\bibliographystyle{plainnat}
\bibliography{references}

\begin{thebibliography}{32}
\providecommand{\natexlab}[1]{#1}
\providecommand{\url}[1]{\texttt{#1}}
\expandafter\ifx\csname urlstyle\endcsname\relax
  \providecommand{\doi}[1]{doi: #1}\else
  \providecommand{\doi}{doi: \begingroup \urlstyle{rm}\Url}\fi

\bibitem[Bache and Lichman(2013)]{Bache+Lichman:2013}
K.~Bache and M.~Lichman.
\newblock {UCI} machine learning repository, 2013.
\newblock URL \url{http://archive.ics.uci.edu/ml}.

\bibitem[Barlett and Cussens()]{barlettadvances}
Mark Barlett and James Cussens.
\newblock Advances in {B}ayesian network learning using integer programming.

\bibitem[Chickering(1996)]{chickering1996learning}
David~Maxwell Chickering.
\newblock Learning {B}ayesian networks is np-complete.
\newblock In \emph{Learning from data}, pages 121--130. Springer, 1996.

\bibitem[Chickering(2003)]{chickering2003optimal}
David~Maxwell Chickering.
\newblock Optimal structure identification with greedy search.
\newblock \emph{The Journal of Machine Learning Research}, 3:\penalty0
  507--554, 2003.

\bibitem[Chickering et~al.(2004)Chickering, Heckerman, and
  Meek]{chickering2004large}
David~Maxwell Chickering, David Heckerman, and Christopher Meek.
\newblock Large-sample learning of {B}ayesian networks is np-hard.
\newblock \emph{The Journal of Machine Learning Research}, 5:\penalty0
  1287--1330, 2004.

\bibitem[Cormen et~al.(2001)Cormen, Leiserson, Rivest, and
  Stein]{Cormen:2001fk}
Thomas~H. Cormen, Charles~E. Leiserson, Ronald~L. Rivest, and Clifford Stein.
\newblock \emph{Introduction to Algorithms}.
\newblock MIT Press, 2001.

\bibitem[Cussens(2012)]{cussens2012bayesian}
James Cussens.
\newblock {B}ayesian network learning with cutting planes.
\newblock \emph{arXiv preprint arXiv:1202.3713}, 2012.

\bibitem[De~Campos et~al.(2009)De~Campos, Zeng, and Ji]{de2009structure}
Cassio~P De~Campos, Zhi Zeng, and Qiang Ji.
\newblock Structure learning of {B}ayesian networks using constraints.
\newblock In \emph{Proceedings of the 26th Annual International Conference on
  Machine Learning}, pages 113--120. ACM, 2009.

\bibitem[de~Campos and Ji(2011)]{de2011efficient}
Cassio~Polpo de~Campos and Qiang Ji.
\newblock Efficient structure learning of {B}ayesian networks using
  constraints.
\newblock \emph{Journal of Machine Learning Research}, 12\penalty0
  (3):\penalty0 663--689, 2011.

\bibitem[De~Campos and Huete(2000)]{de2000new}
Luis~M De~Campos and Juan~F Huete.
\newblock A new approach for learning belief networks using independence
  criteria.
\newblock \emph{International Journal of Approximate Reasoning}, 24\penalty0
  (1):\penalty0 11--37, 2000.

\bibitem[Giudici and Green(1999)]{giudici1999decomposable}
Paolo Giudici and PJ~Green.
\newblock Decomposable graphical gaussian model determination.
\newblock \emph{Biometrika}, 86\penalty0 (4):\penalty0 785--801, 1999.

\bibitem[Heckerman et~al.(1995)Heckerman, Geiger, and
  Chickering]{heckerman1995learning}
David Heckerman, Dan Geiger, and David~M Chickering.
\newblock Learning {B}ayesian networks: The combination of knowledge and
  statistical data.
\newblock \emph{Machine learning}, 20\penalty0 (3):\penalty0 197--243, 1995.

\bibitem[Hemmecke et~al.(2012)Hemmecke, Lindner, and
  Studen{\`y}]{hemmecke2012characteristic}
Raymond Hemmecke, Silvia Lindner, and Milan Studen{\`y}.
\newblock Characteristic imsets for learning {B}ayesian network structure.
\newblock \emph{International Journal of Approximate Reasoning}, 53\penalty0
  (9):\penalty0 1336--1349, 2012.

\bibitem[Jaakkola et~al.(2010)Jaakkola, Sontag, Globerson, and
  Meila]{jaakkola2010learning}
Tommi Jaakkola, David Sontag, Amir Globerson, and Marina Meila.
\newblock Learning {B}ayesian network structure using lp relaxations.
\newblock In \emph{International Conference on Artificial Intelligence and
  Statistics}, pages 358--365, 2010.

\bibitem[Jiang et~al.(2011)Jiang, Neapolitan, Barmada, and
  Visweswaran]{jiang2011learning}
Xia Jiang, Richard Neapolitan, M~Michael Barmada, and Shyam Visweswaran.
\newblock Learning genetic epistasis using {B}ayesian network scoring criteria.
\newblock \emph{BMC bioinformatics}, 12\penalty0 (1):\penalty0 89, 2011.

\bibitem[Knuth(1998)]{Knuth:1998fk}
Donald Knuth.
\newblock \emph{The Art of Computer Programmin}, volume~3, pages 158--168.
\newblock Addison-Wesley, 2nd edition, 1998.

\bibitem[Kollar and Friedman(2009)]{kollar2009probabilistic}
Daphne Kollar and Nir Friedman.
\newblock \emph{Probabilistic graphical models: principles and techniques}.
\newblock The MIT Press, 2009.

\bibitem[Lemire et~al.(2010)Lemire, Kaser, and Aouiche]{lemire2010sorting}
Daniel Lemire, Owen Kaser, and Kamel Aouiche.
\newblock Sorting improves word-aligned bitmap indexes.
\newblock \emph{Data \& Knowledge Engineering}, 69\penalty0 (1):\penalty0
  3--28, 2010.

\bibitem[Madigan and Raftery(1994)]{madigan1994model}
David Madigan and Adrian~E Raftery.
\newblock Model selection and accounting for model uncertainty in graphical
  models using occam's window.
\newblock \emph{Journal of the American Statistical Association}, 89\penalty0
  (428):\penalty0 1535--1546, 1994.

\bibitem[Madigan et~al.(1995)Madigan, York, and Allard]{madigan1995bayesian}
David Madigan, Jeremy York, and Denis Allard.
\newblock {B}ayesian graphical models for discrete data.
\newblock \emph{International Statistical Review/Revue Internationale de
  Statistique}, pages 215--232, 1995.

\bibitem[Moore and Lee(1998)]{Moore:1998fk}
Andrew Moore and Mary~Soon Lee.
\newblock Cached sufficient statistics for efficient machine learning with
  large datasets.
\newblock \emph{Journal of Artificial Intelligence Research}, 8:\penalty0
  67--91, 1998.

\bibitem[Moore and Wong(2003)]{moore2003optimal}
Andrew Moore and Weng-Keen Wong.
\newblock Optimal reinsertion: A new search operator for accelerated and more
  accurate {B}ayesian network structure learning.
\newblock In \emph{ICML}, volume~3, pages 552--559, 2003.

\bibitem[Poess and Potapov(2003)]{poess2003data}
Meikel Poess and Dmitry Potapov.
\newblock Data compression in oracle.
\newblock In \emph{Proceedings of the 29th international conference on Very
  large data bases-Volume 29}, pages 937--947. VLDB Endowment, 2003.

\bibitem[Rincent et~al.(2012)Rincent, Lalo{\"e}, Nicolas, Altmann, Brunel,
  Revilla, Rodriguez, Moreno-Gonzalez, Melchinger, Bauer, Schoen, Meyer,
  Giauffret, Bauland, Jamin, Laborde, Monod, Flament, Charcosset, and
  Moreau]{rincent2012maximizing}
R.~Rincent, D.~Lalo{\"e}, S.~Nicolas, T.~Altmann, D.~Brunel, P.~Revilla, V.~M.
  Rodriguez, J.~Moreno-Gonzalez, A.~Melchinger, E.~Bauer, C-C. Schoen,
  N.~Meyer, C.~Giauffret, C.~Bauland, P.~Jamin, J.~Laborde, H.~Monod,
  P.~Flament, A.~Charcosset, and L.~Moreau.
\newblock Maximizing the reliability of genomic selection by optimizing the
  calibration set of reference individuals: Comparison of methods in two
  diverse groups of maize inbreds (zea mays l.).
\newblock \emph{Genetics}, 192\penalty0 (2):\penalty0 715--728, 2012.

\bibitem[Shang et~al.(2011)Shang, Zhang, Sun, Liu, Ye, and
  Yin]{shang2011performance}
Junliang Shang, Junying Zhang, Yan Sun, Dan Liu, Daojun Ye, and Yaling Yin.
\newblock Performance analysis of novel methods for detecting epistasis.
\newblock \emph{BMC bioinformatics}, 12\penalty0 (1):\penalty0 475, 2011.

\bibitem[Silander and Myllymaki(2012)]{silander2012simple}
Tomi Silander and Petri Myllymaki.
\newblock A simple approach for finding the globally optimal {B}ayesian network
  structure.
\newblock \emph{arXiv preprint arXiv:1206.6875}, 2012.

\bibitem[Singh and Moore(2005)]{singh2005finding}
Ajit~P Singh and Andrew~W Moore.
\newblock Finding optimal {B}ayesian networks by dynamic programming.
\newblock 2005.

\bibitem[Studen{\`y}(2005)]{studeny2005probabilistic}
Milan Studen{\`y}.
\newblock \emph{On Probabilistic Conditional Independence Structures}.
\newblock Springer, 2005.

\bibitem[Studeny and Haws(2011)]{studeny2011polyhedral}
Milan Studeny and David Haws.
\newblock On polyhedral approximations of polytopes for learning bayes nets.
\newblock \emph{arXiv preprint arXiv:1107.4708}, 2011.

\bibitem[Studen\'y and Haws(2013)]{StudenyHaws:2013}
Milan Studen\'y and David Haws.
\newblock Learning {B}ayesian network structure: Towards the essential graph by
  integer linear programming tools.
\newblock \emph{Accepted in Journal of Approximate Reasoning}, 2013.

\bibitem[Studen{\`y} et~al.(2010)Studen{\`y}, Hemmecke, and
  Lindner]{studeny2010characteristic}
Milan Studen{\`y}, Raymond Hemmecke, and Silvia Lindner.
\newblock Characteristic imset: a simple algebraic representative of a
  {B}ayesian network structure.
\newblock In \emph{Proceedings of the 5th European workshop on probabilistic
  graphical models}, pages 257--264. Citeseer, 2010.

\bibitem[Sucheston et~al.(2010)Sucheston, Chanda, Zhang, Tritchler, and
  Ramanathan]{sucheston2010comparison}
Lara Sucheston, Pritam Chanda, Aidong Zhang, David Tritchler, and Murali
  Ramanathan.
\newblock Comparison of information-theoretic to statistical methods for
  gene-gene interactions in the presence of genetic heterogeneity.
\newblock \emph{BMC genomics}, 11\penalty0 (1):\penalty0 487, 2010.

\end{thebibliography}

\section{Appendix}
\label{sec:appendix}
\autoref{alg:qlsrefine} can be easily modified to handle extra input and output 
of more natural data structure to store the current sorting. Moreover, the modification
does not change the running time or space constraints. 

\begin{definition}
Let $D' \in \R^{m \times l}$ be a sub-matrix of $D \in \R^{m \times n}$. 
The unique \emph{order vector} $O \in \R^m$ of some 
ordering $\widehat O$ of rows of $D'$ is a vector such that
$O_i$ is the row index of $D$ of the $i$th item in the ordering $\widehat O$.
\end{definition}

If we take the order $O$ and ranking vector $L$ of the same ordering together
we define another useful data structure.

\begin{definition}
Let $D' \in \R^{m \times l}$ be a sub-matrix of $D \in \R^{m \times n}$
with some ordering $\widehat O$ of rows of $D'$ and its unique order
vector $O$ and ranking vector $L$.
The \emph{partitioning vector} $P \in \R^{\max(L)+1}$ is the vector such that
$P_i$ is the index into the order vector $O$ where the rows of rank $i$
begin.
\end{definition}

The ranking and order vectors store the same information, but it not
necessarily linear time to transform from one to the other. The benefit to the
modified algorithm is that it updates $L$, $O$, and $P$ simultaneously. In many
cases it is easier to work with the ordering vector. Moreover, the partitioning
vector with the order vector gives all the necessary information to form a
contingency table.

\begin{example}
Consider matrix $D$ and the sub-matrix $D''$ in \autoref{ex:D}. The ranking,
ordering, and partitioning vectors are $L'' := [4,3,0,3,2,2,5,4,1,2]^\top$,
$O'' := [2,8,4,5,9,1,3,0,7,6]$, and $P'' := [0,1,2,5,7,9]$. Immediately, we can read off
from $L''$ that there are six rank blocks, if we want to traverse the rows
of the matrix $D''$ in the order we simply use $O''$, and $P''$ tells us how many rows
of each rank are present.
\end{example}

We now give \autoref{alg:qlsrefinenew} which takes the same input as
\autoref{alg:qlsrefine} as well as the ordering and partitioning vectors. It
outputs the new ranking, ordering, and partitioning vector with respect to the
lexicographical ordering one gets by appending column $i$.

\begin{algorithm}[!h]
\begin{algorithmic}[1]
    \REQUIRE $D \in \R^{m\times n}$, $Q \in \Z^{m\times n}$, $i \in \{0,\ldots,n-1\}$, $L \in \Z^m$, $P \in Z^m$.
    \ENSURE $L' \in Z^m$, $O' \in Z^m$, $P' \in Z^m$.
    \STATE $L' := \ve 0 \in Z^m$. 
    \STATE $O' := \ve 0 \in Z^m$. 
    \STATE $P' := \ve P \in Z^m$. 
    \STATE $\IDval := \ve 0 \in Z^m$. \# Records most recent value in $D$ w.r.t.\ ID.
    \STATE $\IDvalInit := \ve 0 \in Z^m$.
    \STATE $subID := \ve 0 \in Z^m$. \# Records subID of $L$.
    \STATE $newCount := \ve 0 \in Z^m$. \# Records count of refinements of each input ID of L.
    \FOR{ $j = 0,\ldots,m-1$ }
        \STATE $O'[P'[L[Q[j,i]]]] = Q[j,i]$.
        \STATE $P'[L[Q[j,i]]] := P'[L[Q[j,i]]] + 1$.
        \IF{ $\IDvalInit[L[Q[j,i]]] == 0$} 
            \STATE $\IDvalInit[L[Q[j,i]]] := 1$.
            \STATE $\IDval[L[Q[j,i]]] := D[Q[j,i],i]$.
        \ELSE
            \IF{ $\IDval[L[Q[j,i]]]\ != D[Q[j,i],i]$}
                \STATE $\IDval[L[Q[j,i]]] := D[Q[j,i],i]$.
                \STATE $newCount[L[Q[j,i]] := newCount[L[Q[j,i]] + 1$.
            \ENDIF
        \ENDIF
        \STATE $subID[Q[j,i]]] := newCount[L[Q[j,i]]]$.
    \ENDFOR
    \STATE $numNewID := 0 \in \Z^M$.
    \STATE $numNewID[m-1] := \sum_{j=0}^{m-2} newCount[j]$.
    \FOR{ $j = m-2,\ldots,1$ }
        \STATE $numNewID[j] := numNewID[j+1] - newCount[j]$.
    \ENDFOR
    \STATE $prevRank := -1$.
    \FOR{ $j = 0,\ldots, m-1$}
        \STATE $L'[O'[j]] := L[O'[j]] + numNewID[L[O'[j]]] + subID[O'[j]]$.
        \IF{ $prevRank <> L'[O[j]]$}
            \STATE $prevRank := L'[O[j]]$.
            \STATE $P'[L'[O[j]]] := j$.
        \ENDIF
    \ENDFOR
    \RETURN $(L',O',P')$.
\end{algorithmic}
\caption{QuickLexSortRefine\dag: Handles order vector.}
\label{alg:qlsrefinenew}
\end{algorithm}

\begin{lemma} 
\label{lem:qlsrefinenew}
If $D \in \R^{m \times n}$, $Q \in \R^{m \times n}$ where the $j$th column of
$Q$ stores the sorting of the $j$th column of $D$, $L$ is the ranking vector,
$O$ is the ordering vector, and $P$ is the partitioning vector of the
lexicographical sorting of the sub-matrix of $D$ determined the sequence of
columns $(a_1,\ldots,a_p)$, and $i \in \{0,\ldots,n-1\}$, then
\autoref{alg:qlsrefinenew} returns the ranking vector $L'$, the ordering vector
$O'$, and the partitioning vector $P'$ of the lexicographical sorting of the
sub-matrix of $D$ determined by the sequence of columns $(a_1,\ldots,a_p,i)$.
\end{lemma}
\begin{proof}
The data structures $IDval$, $IDvalInit$, $newCount$, and $subID$ all are
initialized and updated the same in \autoref{alg:qlsrefinenew} as they were in
\autoref{alg:qlsrefine}. The new order vector $O'$ is initialized to be zeros
and the new partitioning vector $P'$ is initialized to equal the input
partitioning vector $P$. Recall the $P_i$ is the index into the ordering vector
$O$ where the rows of rank $i$ begin. The goal of lines $9-10$ are to create
the new ordering vector $O'$. To do this we must reorder all rows that have the
same previous ranking according to $L$. Thus, the vector $P'$ is temporarily
used to point to the next available index into $O'$ were the rows all have the
same rank according to $L$. In line $9$ we fill in the entries of $O'$ as we
traverse the new column $i$ according to the pre-computed sorting given in
$Q_{\cdot i}$. Specifically we look at the current rank of $Q[j,i]$ which is
given in $L[Q[j,i]]$. Then $P'[L[Q[j,i]]]$ points to the next available
position in $O'$ with rank equal to $L[Q[j,i]]$. Since we have filled this
position $O'$ we increment $P'[L[Q[j,i]]$ in line $10$. In the end, $O'$ will
be the ordering vector with respect to the new order given by the sequence
of columns $(a_1,\ldots,a_p,a_i)$.

In line $27$ we initialize the data structure $prevRank$ which will store the
previously observed rank in the following for loop.  In \autoref{alg:qlsrefine}
and lines $23-24$ we filled in entries of $L'$ by traversing $j=0,\ldots,m-1$.
We note that we could have traversed $j$ in any particular order. Thus, in
\autoref{alg:qlsrefinenew} we traverse $j$ in the order given by the new
ordering vector $O'$.  Therefore, by the arguments in the proof
\autoref{lem:qlsrefine}, $L'$ is the unique ranking vector given by the
sequence of columns $(a_1,\ldots,a_p,a_i)$. 

Lastly in lines $30-32$ whenever we observe a row with a new rank, we set $P'$
to point to the appropriate index into $O'$. Therefore, $P'$ is the 
partitioning vector given by the sequence of columns $(a_1,\ldots,a_p,a_i)$.

\end{proof}

\begin{lemma} 
\autoref{alg:qlsrefinenew} runs in time $O(m)$ and space $O(m)$.
\label{lem:qlsrefinenewruntime}
\end{lemma}
\begin{proof}
There are only three loops in \autoref{alg:qlsrefinenew}, each of them
repeated $m$ times. Each inner operation is constant time. The only space
requirements are determined by column vectors of the $m \times n $ input
matrices and the vectors of length $m$.
\end{proof}

\end{document}